\documentclass[a4paper,UKenglish]{lipics-v2018}

\usepackage{microtype}
\usepackage{paralist}
\usepackage{url}
\usepackage{xcolor}
\usepackage{stmaryrd}
\usepackage{textcomp}
\usepackage{pifont}
\usepackage{bbding}
\usepackage{graphicx}
\usepackage{proof}
\usepackage{turnstile}
\usepackage{cite}
\usepackage{arxiv}

\title{Analysis of Logarithmic Amortised Complexity%
}
\titlerunning{Logarithmic Amortised Complexity}
\author{Martin Hofmann~$\dag$}{Department of Computer Science\\
LMU Munich, Germany}{}{}{}

\author{Georg Moser}{Department of Computer Science \\
University of Innsbruck, Austria}{georg.moser@uibk.ac.at}{http://orcid.org/0000-0001-9240-6128}{Partly supported by DARPA/AFRL contract number FA8750-17-C-088.}
\authorrunning{M.~Hofmann and G.~Moser}

\Copyright{Martin Hofmann and Georg Moser}


\EventEditors{John Q. Open and Joan R. Access}
\EventNoEds{2}
\EventLongTitle{42nd Conference on Very Important Topics (CVIT 2016)}
\EventShortTitle{CVIT 2016}
\EventAcronym{CVIT}
\EventYear{2016}
\EventDate{December 24--27, 2016}
\EventLocation{Little Whinging, United Kingdom}
\EventLogo{}
\SeriesVolume{42}
\ArticleNo{23}
\hideLIPIcs  %uncomment to remove references to LIPIcs series (logo, DOI, ...), e.g. when preparing a pre-final version to be uploaded to arXiv or another public repository
\nolinenumbers
\begin{document}

\maketitle

\begin{abstract}
We introduce a novel amortised resource analysis based on a potential-based
type system. This type system gives rise to logarithmic and polynomial bounds 
on the runtime complexity and is the first such system to exhibit logarithmic
amortised complexity. We relate the thus obtained automatable amortised resource analysis to manual amortised analyses of self-adjusting data structures, like splay trees, that can be found in the literature.
\end{abstract}

\section{Introduction}
In a series of papers a number of researchers including the present
authors, see~\cite{HofmannJ03,HoffmannH10a,HAH11,HAH:2012b,JHLH10,JLHSH09,JLHSH09,HM:2014,HM:2015,MS:2018a}
to name just a few, have explored the area of type-based automated amortised analysis which
lead to several successful tools for deriving accurate bounds on the
resource usage of functional \cite{HAH:2012b,JLHSH09} and also
imperative programs \cite{HR:2013,HS14}, as well as term rewriting \cite{avanzini2016tct,MS:2018a}. 
Accordingly type-based amortised analysis has been employed on a variety of cost metrics. While initially confined to
linear resource bounds \cite{HofmannJ03} the methods were subsequently
extended to cover polynomial \cite{HoffmannH10a}, multivariate polynomial
\cite{HAH:2012b}, and also exponential bounds \cite{HR:2013}.
However, the automated analysis of sublinear, in particular logarithmic resource
bounds remained elusive.

One notable exception is \cite{Nipkow:2015} where
the correct amortised analysis of splay trees\cite{ST:1985,Tarjan:1985} and other data
structures is certified in \isabelle\ with some tactic
support. However, although the analysis is naturally coached in a type system (namely the type system of the proof assistant \textsf{Isabelle}, one cannot speak there of fully automated analysis.
It is in fact not at all clear how a formalisation in an interactive theorem prover as such leads to automation.

It is the purpose of this paper to open up a route towards the
automated, type-based derivation of logarithmic amortised cost. While
we do not yet have a prototype implementation, let alone viable experimental
data, we make substantial progress in that we present a system akin to the multivariate
analysis from \cite{HAH11,HAH:2012b} which reduces the task of justifying a
purported logarithmic complexity bound to the validity of a
well-defined set of inequalities involving linear arithmetic and
logarithmic terms. We also give concrete ideas as to how one can
infer logarithmic bounds efficiently by using an \emph{Ansatz} with unknown
coefficients.

Our analysis is coached in a simple core functional language just sufficiently rich to
provide a full definition of our motivating example: \emph{splaying}. We employ a big-step
semantics, following similar approaches in the literature. However, this implies that our
resource analysis requires termination, which is typically not the case. It is straightforward
to provide a partial big-step semantics~\cite{HoffmannH10b} or a small-step semantics~\cite{MS:2018a}
to overcome this assumption. Furthermore, the proposed type system is geared towards runtime
as computation cost. Again it would not be difficult to provide a parametric type system. We consider
both issues as complementary to our main agenda.

\paragraph*{Organisation}

The rest of this paper is organised as follows. In the next section we introduce a
simple core language underlying our reasoning and provide a full definition of splaying, our
running example. In Section~\ref{Primer} we provide background and a high-level description
of our approach. The employed notion of potential function is provided in Section~\ref{ResourceFunctions},
while our main result is established in Section~\ref{Typesystem}. In Section~\ref{Analysis} we employ the established
type system to splaying, while in Section~\ref{Implementation}, we
clarify the aforementioned \emph{Ansatz} to infer logarithmic bounds. Finally, we conclude in Section~\ref{Conclusion}.

\section{Motivating Example}
\label{Splaying}

In this section, we introduce the syntax of a suitably defined core (first-order) programming language 
to be used in the following. Furthermore, we recall the definition of \emph{splaying}, following the presentation
by Nipkow in~\cite{Nipkow:2015}. Splaying constitutes the motivating examples for the type-based logarithmic
amortised resource analysis presented in this paper.

To make the presentation more succinct, we assume only the following types: 
Booleans ($\Bool = \{\true, \false\}$), an abstract base type $B$, product types, and a type $T$ of binary trees 
whose internal nodes are labelled with elements $\typed{a}{B}$.
Elements $\typed{t}{T}$ are defined by the following grammar which fixes notation. 
\begin{equation*}
  t ::= \nil \mid \tree{t}{a}{t} \tpkt
\end{equation*}
The size of a tree is the number of leaves: $\size{\nil} \defsym 1$, 
$\size{\tree{t}{a}{u}} \defsym \size{t} + \size{u}$.

Expressions are defined as follows and given in \emph{let normal form} to
simplify the presentation of the semantics and typing rules. In order to ease
the readability, we make use of some mild syntactic sugaring in the presentation
of actual code.

\begin{definition}
\begin{align*}
  cmp & ::= \text{\lstinline{<}} \mid \text{\lstinline{>}} \mid
        \text{\lstinline{=}}
  \\
  e & ::= \true\ | \false
    \mid x
    \mid e~cmp~e 
    \mid \cif\ e\ \cthen\ e\ \celse\ e
  \\
   & \mid \vlet\ x~\equal~e\ \vin\ e 
    \mid f(x,\dots,x) 
   \\
   & \mid \match\ x\ \with\ 
      \text{\lstinline{|}} \nil\ \arrow e
      \text{\lstinline{|}} \tree{x}{x}{x}\ \arrow e
\end{align*}
\end{definition}

We skip the standard definition of integer constants $n \in \Z$ as well as variable
declarations, cf.~\cite{Pierce:2002}. Furthermore, we omit binary operations and focus on the bare
essentials for the comparison operators. For the resource analysis these are not of importance,
as long as we assume that no actual costs are emitted.

A \emph{typing context} is a mapping from variables $\VS$ to types. 
Type contexts are denoted by upper-case Greek letters.
A program $\Program$ consists of a signature $\FS$ together with a set of function definitions of the form
$f(x_1,\dots,x_n) = e$, where the $x_i$ are variables and $e$ an expression. 
A \emph{substitution} or (\emph{environment}) $\sigma$ is a mapping from variables to values that respects
types. Substitutions are denoted as sets of assignments: $\sigma = \{x_1 \mapsto t_1, \dots, x_n \mapsto t_n\}$.
We write $\dom(\sigma)$ ($\range(\sigma)$) to denote the domain (range) of $\sigma$. 
Let $\sigma$, $\tau$ be substitutions such that $\dom(\sigma) \cap \dom(\tau) = \varnothing$. Then we
denote the (disjoint) union of $\sigma$ and $\tau$ as $\sigma \dunion \tau$.
We employ a simple cost-sensitive big-step semantics, whose rules are given in Figure~\ref{fig:4}. The judgement
$\eval{\sigma}{m}{e}{v}$ means that under environment $\sigma$, expression $e$ is evaluated to value $v$ in exactly $m$ steps. Here only rule applications emit (unit) costs.

\begin{figure}[t]
\begin{center}
\begin{tabular}{ccc}
  $\infer{\eval{\sigma}{0}{b}{b}}{%
  b \in \{\true,\false\}
  }$
  &
  $\infer{\eval{\sigma}{0}{\nil}{\nil}}{}$
  &
  $\infer{\eval{\sigma}{0}{\tree{x_1}{x_2}{x_3}}{\tree{t}{a}{u}}}{%
      x_1\sigma = t
      &
      x_2\sigma = a
      &
      x_3\sigma = u
  }$
  \\[2ex]
  $\infer{\eval{\sigma}{0}{x}{v}}{%
     x\sigma = v%
     }$  
  &
  $\infer{\eval{\sigma}{m+1}{f(x_1,\ldots,x_k)}{v}}{%
     f(x_1,\ldots,x_k) = e \in \Program
     & \eval{\sigma}{m}{e}{v}
   }$
  &
  $\infer{\eval{\sigma}{0}{x_1~cmp~x_2}{b}}{%
      \text{$b$ is value of $x_1\sigma~cmp~x_2\sigma$}
     }$
  \\[2ex]
\multicolumn{3}{c}{%
  $\infer{\eval{\sigma}{m}{\cif\ x\ \cthen\ e_1\ \celse\ e_2}{v}}{%
     \eval{\sigma}{0}{x}{\true}
     &
     \eval{\sigma}{m}{e_1}{v}
  }$
  \hfill
  $\infer{\eval{\sigma}{m}{\cif\ x\ \cthen\ e_1\ \celse\ e_2}{v}}{%
     \eval{\sigma}{0}{x}{\false}
     &
     \eval{\sigma}{m}{e_2}{v}
  }$
}
  \\[2ex]
  \multicolumn{3}{c}{%
  $\infer{\eval{\sigma}{m}{\vlet\ x~\equal~e_1\ \vin\ e_2}{v}}{%
     \eval{\sigma}{m_1}{e_1}{v'}
     &
     \eval{\sigma[x \mapsto v']}{m_2}{e_2}{v}
     &
     m = m_1 + m_2
  }$}
  \\[2ex]
  \multicolumn{3}{c}{%
  $\infer{\eval{\sigma}{m}{\match\ x~\with\ 
       \begin{array}[t]{l}
         \text{\lstinline{|}} \nil\ \arrow e_1 \\
         \text{\lstinline{|}} \tree{x_1}{x_2}{x_3}\ \arrow e_2
       \end{array}
  }{v}}{%
    x\sigma = \nil
    &
    \eval{\sigma}{m}{e_1}{v}
  }$
  \hfill \hspace{5ex} \hfill
  $\infer{\eval{\sigma}{m}{\match\ x~\with\
      \begin{array}[t]{l}
        \text{\lstinline{|}} \nil\  \arrow e_1 \\
        \text{\lstinline{|}} \tree{x_1}{x_2}{x_3}\ \arrow e_2
      \end{array}
    }{v}}{%
    x\sigma = \tree{t}{a}{u}
    &
    \eval{\sigma'}{m}{e_2}{v}  
  }$}
\end{tabular}
\end{center}
Here $\sigma[x \mapsto v']$ denotes the update of the environment $\sigma$ such that $\sigma[x \mapsto v'](x) = v'$ and the value of all other variables remains unchanged. Furthermore, in the second $\match$ rule, we set $\sigma' \defsym \sigma \dunion \{x_0 \mapsto t, x_1 \mapsto a, x_2 \mapsto u\}$.
\caption{Big-Step Semantics}
\label{fig:4}
\end{figure}

\begin{figure}[t]
\centering
\begin{lstlisting}[mathescape]
splay a t = match t with
    | $\nil$ -> $\nil$
    | $\tree{cl}{c}{cr}$ -> 
       if a = c then $\tree{cl}{c}{cr}$
       else if a < c then match cl with 
          | $\nil$ -> $\tree{cl}{c}{cr}$
          | $\tree{bl}{b}{br}$ -> 
             if a=b then $\tree{bl}{a}{\tree{br}{c}{cr}}$
             else if a<b                        
                  then if bl=$\nil$ then $\tree{bl}{b}{\tree{br}{c}{cr}}$
                       else match splay a bl with
                          | $\tree{al}{a'}{ar}$ -> $\tree{al}{a'}{%
                                                 \tree{ar}{b}{\tree{br}{c}{cr}}%
                                                 }$
                  else if br=$\nil$ then $\tree{bl}{b}{\tree{br}{c}{cr}}$
                       else match splay a br with
                          | $\tree{al}{a'}{ar}$ -> $\tree{\tree{bl}{b}{al}}%
                                                     {a'}%
                                                     {\tree{ar}{c}{cr}}$
            else match cr with 
               | $\nil$ -> $\tree{cl}{c}{cr}$
               | $\tree{bl}{b}{br}$ ->
                   if a=b then $\tree{\tree{cl}{c}{bl}}{a}{br}$
                   else if a<b
                        then if bl=$\nil$ then $\tree{\tree{cl}{c}{bl}}{b}{br}$
                             else match splay a bl with
                                  | $\tree{al}{a'}{ar}$ -> $\tree{\tree{cl}{c}{al}}%
                                                             {a'}%
                                                             {\tree{ar}{b}{br}}$
                        else if br=$\nil$ then $\tree{\tree{cl}{c}{bl}}{b}{br}$
                             else match splay a br with
                                  | $\tree{al}{x}{xa}$ -> $\tree{\tree{%
                                                           \tree{cl}{c}{bbl}%
                                                           }{b}{al}}{x}{xa}$
\end{lstlisting}
\caption{Function \lstinline{splay}.}
\label{fig:1}
\end{figure}

\emph{Splay trees} have been introduced by Sleator and Tarjan~\cite{ST:1985,Tarjan:1985} 
as self-adjusting binary search trees
with strictly increasing inorder traversal. There is no explicit
balancing condition. All operations rely on a tree rotating operation
dubbed \emph{splaying}; \lstinline{splay a t} is performed by rotating element $a$ to the root of tree $t$
while keeping inorder traversal intact. If $a$ is not contained in $t$, then the last element found before
$\nil$ is rotated to the tree. The complete definition is given in Figure~\ref{fig:1}.
Based on splaying, searching is performed by splaying with the sought element and comparing to the root
of the result. Similarly, the definition of insertion and deletion depends on splaying. 
Exemplary the definition of insertion is given in Figure~\ref{fig:2}. 
See also~\cite{Nipkow:2015} for full algorithmic, formally verified, descriptions.

All basic operations can be performed in  $\bO(\log n)$ amortised runtime. The logarithmic amortised
complexity is crucially achieved by local rotations of subtrees in the definition of \lstinline{splay}. 
Amortised cost analysis of splaying has been provided for example by Sleator and Tarjan~\cite{ST:1985},
Schoenmakers~\cite{Schoenmakers93}, Nipkow~\cite{Nipkow:2015}, Okasaki~\cite{Okasaki:1999}, among
others. 
Below, we follow Nipkow's approach, where the actual cost of splaying is measured by counting the
number of calls to $\text{\lstinline{splay}} \colon B \times T \to T$. 

\begin{figure}[t]
\centering
\begin{lstlisting}[mathescape]
insert a t = if t=$\nil$ then $\tree{\nil}{a}{\nil}$
             else match splay a t with
                  | $\tree{l}{a'}{r}$ -> 
                    if a=a' then $\tree{l}{a}{r}$
                    else if a<a' then $\tree{l}{a}{\tree{\nil}{a'}{r}}$
                         else $\tree{\tree{l}{a'}{\nil}}{a}{r}$
\end{lstlisting}
\caption{Function \lstinline{insert}.}
\label{fig:2}
\end{figure}

\begin{figure}[t]
\centering
\begin{lstlisting}[mathescape]
delete a t = if t=$\nil$ then $\nil$
   else match splay a t with
        | $\tree{l}{a'}{r}$ -> 
          if a=a' then if l=$\nil$ then $r$ 
                       else match splay_max l with
                       | $\tree{l'}{m}{r'}$ -> $\tree{l'}{m}{r}$
          else $\tree{l}{a'}{r}$


splay_max t = match t with
   | $\nil$ -> $\nil$
   | $\tree{l}{b}{r}$ -> match r with 
       | $\nil$ -> $\tree{l}{b}{\nil}$
       | $\tree{rl}{c}{rr}$ -> 
         if rr=$\nil$ then $\tree{\tree{l}{b}{rl}}{c}{\nil}$ 
                      else match splay_max rr with
                           | $\tree{rrl}{x}{xa}$ -> $\tree{\tree{\tree{l}{b}{rl}}{c}{rrl}}{x}{xa}$
\end{lstlisting}
\caption{Functions \lstinline{delete} and \lstinline{splay_max}.}
\label{fig:3}
\end{figure}

\section{Background and Informal Presentation}
\label{Primer}

To set the scene we briefly review the general approach up to and
including the multivariate polynomial analysis from Hoffmann et al.~\cite{HAH11,HAH:2012b,HM:2015}. 

\paragraph*{Univariate Analysis}

Suppose that we have types $A,B,C,\dots$ representing sets of values. 
We write $\typedomain{A}$  for the set of values represented by type $A$.  
Types may be constructed from base types by type formers such as list, tree, 
product, sum, etc. 

For each type $A$ we have a, possibly infinite, set of
\emph{basic potential functions} $\BF(A)\colon \typedomain{A} \to \Rplus$. Thus, if $p\in \BF(A)$ and $v\in \typedomain{A}$
then $p(v)\in \Rplus$.  It is often useful to regard $\BF(A)$
as set of \emph{names} for basic potential functions. In this case,
we have a function
$\langle-,-\rangle:\BF(A)\times\typedomain{A}\rightarrow
\Rplus$. To ease notation, one then sometimes writes $p(v)$, instead
of $\langle p,v\rangle$. 

An \emph{annotated type} is a pair of a type $A$ and a function
$Q:\BF(A)\rightarrow \Rplus$ providing a coefficient for each
\emph{basic potential function}. The function $Q$ must be zero on all but
finitely many basic potential functions.
For each annotated type $\annotatedtype{A}{Q}$, the \emph{potential function} $\phi_Q:
\typedomain{A}\rightarrow \Rplus$ is given by
\begin{equation*}
  \phi_Q(v) \defsym \sum_{p\in \BF(A)} Q(p) \cdot p(v) \tpkt
\end{equation*}
Now suppose that we have a function $f \colon A_1 \times\dots\times A_n \rightarrow B$ and
that the actual cost for computing $f(v_1,\dots,v_n)$ is given by $c(v_1,\dots,v_n)$ where $c \colon \typedomain{A_1}\times \dots\times\typedomain{A_n}\rightarrow \Rplus$.
The idea then is to choose annotations $Q_1,\dots,Q_n,Q$ of $A_1,\dots,A_n$ and $B$ in such a way that the amortised cost of $f$ becomes zero or constant, i.e.\
\begin{equation*}
\phi_{Q_1}(v_1) + \dots + \phi_{Q_n}(v_n) \geqslant  c(v_1,\dots,v_n) + \phi_{Q}(f(v_1,\dots,v_n)) + d
\tkom
\end{equation*}
where $d \in Rplus$.
The potential of the input suffices to pay for the cost of
computing $f(v_1,\dots,v_n)$ as well as the potential of the result.
This allows one to compose such judgements in a syntax-oriented way
without having to estimate sizes, let alone the precise form of
intermediate results, which is often needed in competing approaches.

If we introduce product types, we can regard functions with several arguments as unary functions:
Let $A_1\dots,A_n$ be types, then so is $A_1 \times\dots\times A_n$. We conclusively define
$\typedomain{A_1 \times\dots\times A_n} \defsym \typedomain{A_1} \times\dots\times \typedomain{A_n}$ and 
$\BF(A_1 \times\dots\times A_n) = \BF(A_1) \dunion \cdots \dunion \BF(A_n)$,
where $\dunion$ stands for disjoint union.
Furthermore, we define $\langle \mathrm{in}_i(p), (v_1,\dots,v_n)\rangle \defsym \langle p,v_i\rangle$. 
If we now regard $f$ above as a unary function from $A_1
\times\dots\times A_n$ to $B$ then it is not hard to see that the
notions of annotation and amortised cost agree with the multi-ary ones
given above.

This approach has been key to lift earlier results on automated resource analysis, 
e.g.~\cite{HofmannJ03}, restricted to \emph{linear} bounds to \emph{polynomial} bounds. 
In particular, in~\cite{HoffmannH10a} an
automated amortised resource analysis has been introduced exploiting these idea. This analysis
employs binomial coefficients as basic potential functions. The approach generalises to 
general inductive data types, cf.~\cite{HM:2014,HM:2015,MS:2018a}.

\paragraph*{Multivariate Analysis}

In the multivariate version of automated amortised analysis \cite{HAH11,HAH:2012b,HM:2015} one takes 
a more general approach to products. Namely, one then puts
\begin{align*}
\BF(A_1 \times\dots\times A_n) & \defsym \BF(A_1) \times\dots\times \BF(A_n)\\
\langle(p_1,\dots,p_n),(v_1,\dots,v_n)\rangle & \defsym \langle p_1,v_1\rangle\cdot \dots \cdot \langle p_n,v_n\rangle
\tkom
\end{align*}
i.e.\ the basic potential function for a product type is obtained as the multiplication of the
basic potential functions of its constituents. In order to achieve backwards compatibility, that is,
to recover all the potential functions available in the univariate case, it is 
necessary to postulate for each type $A$ a distinguished element $1\in \BF(A)$ 
with $\langle 1,a\rangle =1$ for all $a\in\typedomain{A}$. 

Consider automatisation of the univariate or multivariate analysis. Suppose that it is possible to derive amortised
costs for basic functions like constructors, if-then-else etc. Then one
sets up annotations with indeterminate coefficients and solves
for them so as to automatically infer costs. This is in particular
possible when the basic potential functions for datatypes like lists
or trees are polynomial functions of length and other size
parameters. One of the reasons why this works so well is that if
$p(n)$ is a polynomial, so is $p(n+1)$ and in fact can be expressed as
a linear combination of basic polynomials like, e.g., powers of $x$ or
binomial coefficients, cf.~\cite{HAH11,HAH:2012b}. This approach also generalises
to general inductive data types, cf.~\cite{HM:2014,HM:2015,MS:2018a}.

In contrast to the univariate system, the multivariate system provides for greater 
accuracy because it can derive bounds like $mn$ which in the univariate analysis would be
over-approximated by $m^2+n^2$. This, however, requires a more careful
management of variables and contexts resulting in rather involved
typing rules for composition (\lstinline{let}) and sharing, where \emph{sharing} refers
to the multiple use of variables. 

Since we need a similar mechanism in the present system
we will explain this in a little more detail.
Let $f \colon A\rightarrow B$ and $g \colon B\times C\rightarrow D$ be functions and suppose that evaluating
$f(x)$ and $g(y,z)$ incurs costs $c(x)$ and $d(y,z)$, respectively. 
Suppose further the following constraints hold for all
$x\in\typedomain{A}$, $y\in\typedomain{B}$, $z\in\typedomain{D}$ and
potential functions $\phi_i$, $\phi_i'$, $\phi_i''$, $\psi$:
\begin{align}
\phi_0(x) &\geqslant c(x) + \phi'_0(f(x)) 
\label{AARA:I}
  \\
\phi_i(x) &\geqslant \phi'_i(f(x)) \quad \text{for all $i$ ($0 < i \leqslant n$)} 
\label{AARA:II}
  \\
\phi'_0(y)+\sum_{i=1}^n \phi'_i(y)\phi_i''(z) &\geqslant d(y,z) + \psi(g(y,z))
\label{AARA:III}
\tpkt
\end{align}
Then we conclude for all $x,y,z$: 
$\phi_0(x)+\sum_{i=1}^n \phi_i(x)\phi_i''(z) \geqslant c(x)+d(g(f(x),y)) + \psi(g(f(x),y))$
guaranteeing that a suitable combination of the potential of the arguments, suffices to pay for the
cost $c(x)$ of computing $f(x)$, the cost $d(g(f(x),y))$ of the function composition $g(f(x),y)$,
as well as for the potential $\psi(g(f(x),y))$ of the result $g(f(x),y)$.
Here we multiply~\eqref{AARA:II} with $\phi''_i(z)$ for $i=1\dots n$.

We emphasise that this requires the possibility of deriving inequalities like~\eqref{AARA:II},
which only involve potentials, but no actual costs. That is, in the multivariate
case we crucially employ a \emph{cost-free semantics} to handle composition of
functions. In a cost-free semantics the whose evaluation does not emit any costs.

\paragraph*{Logarithmic Amortised Costs} 

We can now explain at this high level the main ingredients of the proposed amortised resource
analysis for logarithmic amortised costs, which also provides some intuition for the type system
established in Section~\ref{Typesystem}. Among other potential functions which we introduce later,
we use (linear combinations of) functions of the form
\begin{equation*}
p_{a_1,\dots,a_n,b}(x_1,\dots,x_n)=\log(a_1x_1+\dots+a_nx_n+b)
\tkom  
\end{equation*}
where $a_1,\dots a_n,b \in \N$. We then have
$p_{a_1,\dots,a_n,b}(x_1+1,x_2,\dots,x_n) = p_{a_1,\dots,a_n,b+a_1}(x_1,\dots,x_n)$,
which constitutes the counterpart of the fact that shifts of
polynomials are themselves polynomials. Similarly, 
$p_{a_0,a_1,\dots,a_n,b}(x_1,x_1,x_2,\dots,x_n) = p_{a_0+a_1,\dots,a_n,b}(x_1,\dots,x_n)$,
which forms the basis of sharing. For composition we use the following
reasoning. Suppose that
\begin{align}
\phi_0(x) &\geqslant c(x) + \phi'_0(f(x)) 
\label{LARA:I}
\\
\log(a_i\size{x}+b_i) &\geqslant \log(a'_i\size{f(x)}+b'_i) \quad \text{for all $i$ ($0 < i \leqslant n$)}
\label{LARA:II}
\\
\phi'_0(y)+\sum_{i=1}^n \log(a'_i\size{y}+a''_i \size{z} +b'_i) &\geqslant d(y,z) + \psi(g(y,z))
\tkom
\label{LARA:III}
\end{align}
where $\size{\cdot}$ are arbitrary
nonnegative functions, and the potential functions $\phi,\psi$ are as before. 
Then we can conclude, arguing similarly as in the multivariate case, that the
following inequality holds:
\begin{equation*}
  \phi_0(x)+\sum_{i=1}^n \log(a_i\size{x}+a''_i \size{z}+b_i) \geqslant c(x)+d(f(x),y) + \psi(g(f(x),y))
\tpkt
\end{equation*}
Here we crucially use strict monotonicity of the logarithm function, in particular the fact that
$\log(u) \geqslant \log(v)$ implies $\log(u+w)\geqslant\log(v+w)$ for $u,v,w\geqslant 1$, cf.~Lemma~\ref{l:3}. Again the
potential of the arguments suffices to pay for the cost of computing $c(x)$, $d(f(x),y)$, respectively
and covers in addition the potential of the result. 

We emphasise the crucial use of \emph{cost-free semantics} for the correct analysis of function
composition, as witnessed by constraint~\eqref{LARA:II}.

\section{Resource Functions}
\label{ResourceFunctions}

In this section, we detail the basic potential functions employed and clarify the definition
of potentials used.

Only trees are assigned non-zero potential. This is not a severe restriction as potentials for basic datatypes
would only become essential, if the construction of such types would emit actual costs. This is not the case
in our context. Moreover, note that list can be conceived as trees of particular shape. 
The potential $\Phi(t)$ of a tree $t$ is given as a non-negative linear combination of basic functions,
which essentially amount to ``sums of logs'', cf.~Schoenmakers~\cite{Schoenmakers93}.
It suffices to specify the basic functions for the type of trees $T$.
More precisely, the \emph{rank} $\rk(t)$ of a tree is defined as follows:
\begin{align*}
  \rk(\nil) &\defsym 0\\
  \rk(\tree{t}{a}{u}) &\defsym \rk(t) + \log'(\size{t}) + \log'(\size{u}) + \rk(u)                        
\tpkt
\end{align*}
% (i) $\rk(\nil) \defsym 0$, 
% (ii) $\rk(\tree{t}{a}{u}) \defsym \rk(t) + \log'(\size{t}) + \log'(\size{u}) + \rk(u)$.
%
Here $\log'(n) \defsym \log_2(\max\{n,1\})$, such that the (binary) 
logarithm function is defined for all numbers. 
This is merely a technicality, introduced to ease the presentation.
Furthermore, recall that $\size{t}$ denotes the number of leaves in tree $t$. 
In the following, we will denote the modified logarithmic 
function, simply as $\log$.
\label{d:log'}
The definition of ``rank'' is inspired by the definition of potential in~\cite{Schoenmakers93,Nipkow:2015},
but subtly changed to suit it to our context.

\begin{definition}
The \emph{basic potential functions} of $T$ are either
\begin{itemize}
\item $\lambda t. \rk(t)$, or
\item $p_{(a,b)} \defsym \log(a \cdot \size{t} +b )$, where $a,b$ are numbers.
\end{itemize}
\end{definition}

The basic functions are denoted as~$\BF$. Note that the
constant function $1$ is representable: $1 = \log(0 \cdot \size{t}+2)$.

% \begin{verbatim} 
% Hi  Georg, ich glaube  uebrigens, dass, wenn  man als
% Basic  Potential  Funktion  neben  rk(x) und  log(ax+b)  auch  noch  x
% (lineare Funktion)  erlaubt, man dann auch  mergesort analysieren kann
% und zwar so,  dass n log n rauskommt (also  a*rk(x)+b). Das waere doch
% dann wirklich eine schoene Story. --martin
%  \end{verbatim}
% TODO: lineares potential in TR formulieren und mergesort skizzieren

Following the recipe of the high-level description in Section~\ref{Primer}, potentials
or more generally \emph{resource functions} become definable as linear combination of basic potential functions.

\begin{definition}
A \emph{resource function} $r \colon \typedomain{T} \to \Rplus$ is a non-negative
linear combination of basic potential functions, that is,
\begin{equation*}
  r(t) \defsym \sum_{i \in \N} q_{i} \cdot p_i (t) \tkom
\end{equation*}
where $p_i \in \BF$. The set of resource functions is denoted as $\RF$.
\end{definition}

We employ $\ast$, natural numbers $i$ and pairs of natural numbers $(a,b)_{a,b \in \N}$ as
indices of the employed basic potential functions. 
A \emph{resource annotation over $T$}, or simply \emph{annotation}, is a 
sequence $Q = [q_\ast] \cup [(q_{(a,b)})_{a,b \in \N}]$ with $q_\ast, q_{(a,b)} \in \Qplus$ with all but
finitely many of the coefficients $q_\ast, q_{(a,b)}$ equal to 0.  It represents
a (finite) linear combination of basic potential functions, that is, a resource function. 
The empty annotation, that is, the annotation 
where all coefficient are set to zero, is denoted as $\varnothing$.

\begin{remark}
We use the convention that the sequence elements of resource annotations 
are denoted by the lower-case letter of the annotation, 
potentially with corresponding sub- or superscripts.
\end{remark}

\begin{definition}
The \emph{potential} of a tree $t$ with respect to an annotation $Q$,
that is, $Q=[q_\ast] \cup [(q_{(a,b)})_{a,b \in \N}]$, is defined as follows.
\begin{equation*}
  \potential{t}{Q} \defsym q_\ast \cdot \rk(t) + \sum_{a,b \in \N} q_{(a,b)} \cdot p_{(a,b)}(t) \tkom
\end{equation*}
Recall that $p_{(a,b)} = \log(a \cdot \size{t} +b)$ and that $\rk$ is the rank
function, defined above.
\end{definition}

\begin{example}
Let $t$ be a tree, then it's potential could be defined as follows: 
$\rk(t) + 3 \cdot \log(\size{t}) + 1$. With respect to the above definition this potential
becomes representable by setting $q \defsym 1, q_{1,0} \defsym 3, q_{0,2} \defsym 1$.
Conclusively $\potential{t}{Q} = \rk(t) + 3 \cdot \log(\size{t}) + 1$. 
\qed
\end{example}

We emphasise that the linear combination defined above is not independent. Consider, for
example $\log(2\size{t}+2)=\log(\size{t}+1)+1$. 

The potential of a sequence of trees $t_1,\dots,t_m$ is defined
as the linear combination of $\rk(t_i)$ and a straightforward extension of the basic
potential functions $p_{a,b}$ to $m$ arguments, denoted as $p_{(\seq[m]{a},b)}$. Here
$p_{(\seq[m]{a},b)}(t_1,\dots,t_m)$ is defined as the logarithmic function
$\log(a_1 \cdot \size{t_1} + \dots + a_m \cdot \size{t_m} + b)$, where
$a_1,\dots,a_m,b \in \N$. 

More precisely, we first generalise annotations to sequences of trees.  
An annotation for a sequence of length $m$ is a sequence $Q = [q_1,\dots,q_m] \cup
[(q_{(\seq[m]{a},b)})_{a_i \in \N}]$, again vanishing almost everywhere. 
Note that an annotation of length $1$ is simply an annotation, where the coefficient $q_1$ is set
equal to the coefficient $q_\ast$.
Based on this, the potential of $t_1,\dots,t_m$ is defined as follows.

\begin{definition}
\label{d:potential}
Let $t_1,\dots,t_m$ be trees and let 
$Q = [q_1,\dots,q_m] \cup [(q_{(\seq[m]{a},b)})_{a_i \in \N}]$ 
be an annotation of length $n$ as above. We define
\begin{equation*}
  \potential{t_1,\dots,t_m}{Q} \defsym \sum_{i=1}^m q_i \cdot \rk(t_i) + \sum_{\seq[m]{a},b \in \N} 
q_{(\seq[m]{a},b)} \cdot p_{(\seq[m]{a},b)}(t_1,\dots,t_m)
  \tkom
\end{equation*}
where $p_{(\seq[m]{a},b)}(t_1,\dots,t_m) \defsym \log(a_1 \cdot \size{t_1} + \cdots + a_m \cdot \size{t_m} + b)$.
Note that $\potential{\varnothing}{Q} = \sum_{b \in \N} q_b \log(b)$.  
\end{definition}

Let $t$ be a tree. Note that the rank function $\rk(t)$ amounts to the sum of the logarithms of the
size of subtrees of $t$. In particular if the tree $t$ simplifies to a list of length $n$, then $\rk(t) = \sum_{i=1}^n \log(i)$.
Moreover, as $\sum_{i=1}^n \log(i) \in \Theta(n \log n)$, the above defined potential functions are sufficiently rich
to express linear combinations of sub- and super-linear functions. 
For practical purposes it may be necessary to expand the class of potential functions further. Here, we emphasise that 
it is not difficult to see the basic potential functions $p_{\seq[m]{a},b}$ can be generalised as
to also incorporate \emph{linear} dependencies on the size of arguments; this does not invalidate
any of the results in this section.  

Let $\sigma$ denote a substitution, let $\Gamma$ denote a typing context 
and let $\typed{x_1}{T},\dots,\typed{x_m}{T}$ denote
all tree types in $\Gamma$. A \emph{resource annotation for $\Gamma$} or simply \emph{annotation}
is an annotation for the sequence of trees $x_1\sigma, \dots, x_m\sigma$.  
We define the \emph{potential} of $\potential{\Gamma}{Q}$ with respect to $\sigma$ as
$\spotential{\Gamma}{Q} \defsym \potential{x_1\sigma,\dots,x_m\sigma}{Q}$.

\begin{definition}
An \emph{annotated signature} $\overline{\FS}$
is a mapping from functions $f$ to sets of pairs consisting of the annotation type
for the arguments of $f$ $\annotatedtype{A_1 \times \cdots \times A_n}{Q}$
and the annotation type $\annotatedtype{A'}{Q'}$ for the result:
\begin{equation*}
  \overline{\FS}(f) \defsym \left\{ \atypdcl{A_1 \times \cdots \times A_n}{Q}{A'}{Q'} \colon \text{\parbox{50ex}{if $f$ takes $m$ trees as arguments, $Q$ is an annotation of length $m$ and $Q'$ a resource annotation}} \right\} \tpkt
\end{equation*}
Note that $m \leqslant n$ by definition.
\end{definition}

We confuse the signature and the annotated signature and denote
the latter simply as~$\FS$. Instead of $\atypdcl{A_1 \times \cdots \times A_n}{Q}{A'}{Q'} \in \FS(f)$, we
typically write $\typed{f}{\atypdcl{A_1 \times \cdots \times A_n}{Q}{A'}{Q'}}$.
%
% GM: implementation detail
As our analysis makes use of a \emph{cost-free semantics} any function symbol is possibly equipped with a \emph{cost-free} signature, independent of~$\FS$. The cost-free signature is denoted as $\FScf$.

\begin{example}
Consider the function \lstinline{splay}: $B \times T \to T$. The induced
annotated signature is given as $\atypdcl{B \times T}{Q}{T}{Q'}$,
where $Q \defsym [q_\ast] \cup [(q_{(a,b)})_{a,b \in \N}]$ and 
$Q' \defsym [q'_\ast] \cup [(q'_{(a,b)})_{a,b \in \N}]$.
The logarithmic amortised cost of splaying, is then
suitably expressed through the following setting:
$q_\ast \defsym 1$, $q_{(1,0)} = 3$, $q_{(0,2)} = 1$, 
$q'_\ast \defsym 1$.
All other coefficients are zero.

This amounts to a potential
of the arguments $\rk(t) + 3 \log(\size{t}) + 1$, while for the
result we consider only its rank.
The correctness of the induced logarithmic amortised costs for splaying is
verified in Section~\ref{Analysis}. 
\qed
\end{example}

Suppose $\potential{t_1,\dots,t_n,u_1,u_2}{Q}$ denotes an annotated sequence of length $n+2$. 
Suppose $u_1 = u_2$ and we want to \emph{share} the values $u_i$, that is, the corresponding function arguments
appear multiple times in the body of the function definition. Then we make use of
the operator $\share{Q}$ that adapts the potential suitably. The operator
is also called \emph{sharing operator}.

\begin{lemma}
\label{l:2}
Let $t_1,\dots,t_n,t,u_1,u_2$ denote a sequence of trees of length $n+2$ with
annotation $Q$. Then there exists a resource annotation $\share{Q}$ such that
$\potential{t_1,\dots,t_n,u_1,u_2}{Q} = \potential{t_1,\dots,t_n,u}{\share{Q}}$, if
$u_1 = u_2 = u$.
\end{lemma}
\begin{proof}
Wlog.\ we assume $n=0$. Thus, let $Q = \{q_1,q_2\} \cup \{(q_{(a_1,a_2,b)})_{a_i \in \N}\}$. 
By definition 
\begin{equation*}
  \potential{u_1,u_2}{Q} = q_1 \cdot \rk(u_1) + q_2 \cdot \rk(u_2) + 
  \sum_{a_1,a_2,b \in \N}  q_{(a_1,a_2,b)} \cdot p_{(a_1,a_2,b)}(u_1,u_2) 
  \tkom
\end{equation*}
where $p_{(a_1,a_2,b)}(u_1,u_2) = \log(a_1 \cdot \size{u_1} + a_2 \cdot \size{u_2} + b)$.
By assumption $u=u_1=u_2$. Thus, we obtain
\begin{align*}
  \potential{u,u}{Q} & = q_1 \cdot \rk(u) + q_2 \cdot \rk(u) + \sum_{a_1,a_2,b \in \N}  q_{(a_1,a_2,b)} \cdot p_{(a_1,a_2,b)}(u,u) 
  \\
  & = (q_1+q_2) \rk(u) + \sum_{a_1+a_2,b \in \N} q_{(a_1+a_2,b)} \cdot p_{(a_1+a_2,b)}(u)
  \\
  &  = \potential{u}{\share{Q}}
    \tkom
\end{align*}
for suitable defined annotation $\share{Q}$.
\end{proof}

We emphasise that the definability of the sharing annotation $\share{Q}$ is based on the fact that
the basic potential functions $p_{\seq[m]{a},b}$ have been carefully chosen so that
\begin{equation*}
  p_{a_0,a_1,\dots,a_m,b}(x_1,x_1,\dots,x_m) = p_{a_0+a_1,\dots,a_m,b}(x_1,x_1,\dots,x_m)
  \tkom
\end{equation*}
%$p_{a_0,a_1,\dots,a_m,b}(x_1,x_1,\dots,x_m) = p_{a_0+a_1,\dots,a_m,b}(x_1,x_1,\dots,x_m)$
holds, cf.~Section~\ref{Primer}.

\begin{remark}
  We observe that the proof-theoretic analogue of the sharing operation constitutes in a
  contraction rule, if the type system is conceived as a proof system.   
\end{remark}

Let $Q = [q_\ast] \cup [(q_{(a,b)})_{a,b \in \N}]$ be an annotation and let 
$K \in \Qplus$. Then we define $Q' \defsym Q + K$ as follows: 
$Q' = [q_\ast] \cup [(q'_{(a,b)})_{a,b \in \N}]$,
where $q'_{(0,2)} \defsym q_{(0,2)} + K$ and for all $(a,b) \not= (0,2)$
$q_{(a,b)}' \defsym q_{(a,b)}$. By definition the annotation coefficient $q_{(0,2)}$ is
the coefficient of the basic potential function $p_{(0,2)}(t) = \log(0\size{t} + 2) = 1$, so
the annotation $Q+K$, adds cost $K$ to the potential induced by $Q$.

Due to the involved form of the basic function underlying the definition of
potential, cf.~Definition~\ref{d:potential}, we cannot simply define \emph{weakening}
of potentials through the (pointwise) comparison of annotations. This is in 
contrast to results on resource analysis for constant amortised costs. 
Instead we compare potentials \emph{symbolically} by fixing the shape of the
considered logarithmic functions and perform coefficient comparisons, akin to similar techniques
in the synthesis of \emph{polynomial interpretations}~\cite{contejean:2005}. In addition
we use basic laws of the $\log$ functions as well as properties of the size function.

Let $\Gamma$ denote a type context containing the type declarations $\typed{x_1}{T}, \dots, \typed{x_m}{T}$ 
and let $Q$ be an annotation of length $m$. The the \emph{symbolic potential}, denoted
as $\potential{\Gamma}{Q}$, is defined as:
\begin{equation*}
\potential{x_1,\dots,x_m}{Q} \defsym \sum_{i=1}^m q_i \cdot \rk(x_i) + \sum_{\seq[m]{a},b \in \N} 
q_{(\seq[m]{a},b)} \cdot p_{(\seq[m]{a},b)}(x_1,\dots,x_m)
\tkom
\end{equation*}
where $p_{(\seq[m]{a},b)}(x_1,\dots,x_m) = \log(a_1 \cdot \size{x_1} + \cdots + a_m \cdot \size{x_m} + b)$.

In order to automate the verification of the constraint $\potential{\Gamma}{Q} \geqslant \potential{\Gamma}{Q'}$,
we can rely on a suitably defined heuristics based on the following simplification steps:
%
% \begin{inparaenum}[(i)]
\begin{enumerate}
\item  simplifications, like e.g.~$\rk(\tree{u}{b}{v}) = \rk(u)+\rk(v)+\log(\size{u})+log(\size{v})$;
\item monotonicity of $\log$;
\item simply estimations of the logarithm functions like the next lemma; and
\item properties of the size function.
\end{enumerate}  
%\end{inparaenum}
\label{d:heuristics}

\begin{lemma}
\label{l:1}
Let $x, y \geqslant 1$. Then $2 + \log(x) + \log(y) \leqslant 2\log(x+y)$. 
\end{lemma}
\begin{proof}
We observe
  \begin{equation*}
    (x+y)^2 -4xy = (x-y)^2 \geqslant 0
    \tpkt
  \end{equation*}
Hence $(x+y)^2 \geqslant 4xy$ and from the monotonicity
of $\log$ we conclude $\log(xy) \leqslant \log(\frac{(x+y)^2}{4})$. 
By elementary laws of $\log$ we obtain:
\begin{equation*}
  \log(\frac{(x+y)^2}{4}) = \log\left( (\frac{x+y}{2})^2 \right) =
  2 \log(x+y) - 2
  \tkom
\end{equation*}
from which the lemma follows as $\log(xy) = \log(x)+\log(y)$.
\end{proof}

We leave the simple proof to the reader. A variant of this fact has already been observed by 
Okasaki, cf.~\cite{Okasaki:1999}. 
The above heuristic is automatable, employing off-the-shelf SMT solvers, such that the required simplification rules
are incorporated as (user-defined) axioms. However, this is not very efficient. In Section~\ref{Implementation} we sketch an alternative path towards automation.

\section{Logarithmic Amortised Resource Analysis}
\label{Typesystem}

In this section, we present the central contribution of this work. We delineate
a novel type system incorporating a potential-based amortised resource analysis
capable of expressing \emph{logarithmic} amortised costs. Soundness of the approach
is established in Theorem~\ref{t:1}.

The next auxiliary lemma is a direct consequence of the
strict monotonicity of $\log$. Note, that the assumption that
$a,b,c$ are strictly greater than is zero is necessary, even in the
light of our use of a ``modified'' logarithm function, see~page~\pageref{d:log'}.

\begin{lemma}
\label{l:3}
Let $u, v, w \geqslant 1$. If 
$\log(u) \leqslant \log(v)$, then
$\log(u+w) \leqslant \log(v+w)$.
\end{lemma}

From the lemma we conclude for coefficients $q_i$ and positive rational number $\seq{q}$, $b$ and $c$, that we have:
\begin{equation*}
  \sum_{i} q_i \cdot \log(a_i) \geqslant \log(b) \ \text{implies} \
  \sum_{i} q_i \cdot \log(a_i + c) \geqslant \log(b+c)
  \tpkt
\end{equation*}
The above inequality is employed in the correct assessment of the transfer of potential in the case of
function composition, see Figure~\ref{fig:5} as well as the high-level description provided in Section~\ref{Primer}.

Our potential-based amortised resource analysis is coached in a type system, 
which is given in Figure~\ref{fig:5}. If the type judgement
$\tjudge{\Gamma}{Q}{e}{A}{Q'}$ is derivable, then the cost of execution of the
expression $e$ is bound from above by the difference between the potential $\spotential{\Gamma}{Q}$
before the execution and the potential $\potential{v}{Q'}$ of the value $v$ obtained through the evaluation of
the expression $e$. The typing system makes use of a \emph{cost-free}
semantics, which does not attribute any costs to the calculation.
I.e.\ the $(\app)$ is changed as no cost is emitted. The cost-free typing judgement
is denoted as $\tjudgecf{\Gamma}{Q}{e}{A}{Q'}$. 

\begin{remark}
Principally the type system can be parametrised
in the resource metric (see~eg.\cite{HAH:2012b}). However, we focus on worst-case runtime complexity, 
symbolically measured through the number of rule applications. 
\end{remark}

We consider the typing rules in turn; recall the convention that sequence elements of annotations are denoted by the lower-case letter of the annotation.
The variable rule $(\var)$ types a variable of unspecified
type $A$. As no actual costs are required the annotation is unchanged.
Similarly no resources are lost through the use of control operators. Conclusively, the definition
of the rules $(\cmp)$ and $(\ite)$ is straightforward.

As exemplary constructor rules, we have rule $(\leafrule)$ for the
empty tree and rule $(\noderule)$ for the node constructor. Both rules
define suitable constraints on the resource annotations to guarantee that the potential of the
values is correctly represented.

The application rule $(\app)$ represents the application of a rule 
in~$\Program$; the required annotations for the typing context and the result can be directly
read off from the annotated signature. Each application emits actual cost $1$, which
is indicated in the addition of $1$ to the annotation $Q$. 

In the pattern matching rule $(\match)$ the potential freed through the
destruction of the tree construction is added to the annotation $R$, which
is used in the right premise of the rule. Note that the length of
the annotation $R$ is $m+2$, where $m$ equals the number of tree types in the type context $\Gamma$.

The constraints expressed in the typing rule $\letrule$, 
guarantee that the potential provided
through annotation $Q$ is distributed among the call to $e_1$ and $e_2$. This typing rule
takes care of function composition. Due to the sharing rule, we can assume
wlog.\ that each variable in $e_1$ and $e_2$ occurs at most once.
The numbers $m$, $k$, respectively, denote the number of tree types in $\Gamma$, $\Delta$.
This rule necessarily employs the \emph{cost-free semantics}. The premise
$\tjudgecf{\Gamma}{P_{\vecb}}{e_1}{A}{P'_{\vecb}}$ ($\vecb \not= \vec{0}$) expresses
that for all non-zero vectors $\vecb$, the potentials $\potential{\Gamma}{P_{\vecb}}$ suffices
to cover the potential $\potential{A}{P'_{\vecb}}$, if not extra costs are emitted. Intuitively this
represents the cost-free constraint \eqref{LARA:II} emphasised in Section~\ref{Primer}.
 
Finally, the type system makes use of structural rules, like the \emph{sharing rule} 
$(\sharerule)$ and the \emph{weakening rules} $(\weakvar)$ and $(\weak)$. The
sharing rule employs the sharing operator, implicitly defined in Lemma~\ref{l:2}. Note that
the variables $x,y$ introduced in the assumption of the typing rule are fresh variables, that do not occur
in $\Gamma$. The weakening rules embody changes in the potential of the type context of
expressions considered. Weakening employs the symbolic potential expressions, introduced
in Section~\ref{ResourceFunctions}.

\begin{figure}[t]
\begin{center}
\begin{tabular}{c@{}c}
  $\infer[(\leafrule)]{\tjudge{\varnothing}{Q}{\nil}{T}{Q'}}{%
    q_{c} = \sum_{a+b=c} q'_{a,b}
  }$
  &
  $\infer[(\weakvar)]{\tjudge{\Gamma, \typed{x}{A}}{Q}{e}{C}{Q'}}{%
      \tjudge{\Gamma}{R}{e}{C}{Q'}
      &
        r_i = q_i
      &
        r_{\veca,b} = q_{\veca,0,b}
       }$
  \\[3ex]
  \multicolumn{2}{c}{%
  $\infer[(\noderule)]{\tjudge{\typed{x_1}{T},\typed{x_2}{B},\typed{x_3}{T}}{Q}{\tree{x_1}{x_2}{x_3}}{T}{Q'}}{%
    q_1 = q_2 = q'
    &
    q_{(1,0,0)} = q_{(0,1,0)} = q'_\ast
    &
    q_{(a,a,b)} = q'_{(a,b)}
  }$}
  \\[3ex]
  $\infer[(\cmp)]{\tjudge{\typed{x_1}{B}, \typed{x_2}{B}}{Q}{x_1~cmp~x_2}{B}{Q}}{%
     \text{$cmp$ a comparison operator}
  }$
  &
  $\infer[(\ite)]{\tjudge{\Gamma, \typed{x}{Bool}}{Q}{\cif\ x\ \cthen\ e_1\ \celse\ e_2}{A}{Q'}}{%
    \tjudge{\Gamma}{Q}{e_1}{A}{Q'} 
    &
    \tjudge{\Gamma}{Q}{e_2}{A}{Q'}   
    }$
  \\[3ex]
  \multicolumn{2}{c}{%
  $\infer[(\matchrule)]{\tjudge{\Gamma, \typed{x}{T}}{Q}{\match\ x~\with\ 
    \text{\lstinline{|}} \nil\ \arrow e_1
    \text{\lstinline{|}} \tree{x_1}{x_2}{x_3} \arrow e_2}{A}{Q'}}{%
    \begin{minipage}[b]{25ex}
      $r_{(\veca,a,a,b)} = q_{(\veca,a,b)}$\\[1ex]
      $p_{\veca,c} = \sum_{a+b=c} q_{\veca,a,b}$\\[1ex]
      $\tjudge{\Gamma}{P}{e_1}{A}{Q'}$
    \end{minipage}
  &
    \begin{minipage}[b]{40ex}
       $r_{m+1} = r_{m+2} = q_{m+1}$
      \\[1ex]
      $r_{(\vec{0},1,0,0)} = r_{(\vec{0},0,1,0)} = q_{m+1}$
      \\[1ex]
      $\tjudge{\Gamma, \typed{x_1}{T}, \typed{x_2}{B}, \typed{x_3}{T}}{R}{e_2}{A}{Q'}$
    \end{minipage}
  }$}
  \\[3ex]
  \multicolumn{2}{c}{%
  $\infer[(\letrule)]{\tjudge{\Gamma, \Delta}{Q}{\vlet\ x~\equal~e_1\ \vin\ e_2}{C}{Q'}}{%
  \begin{minipage}[b]{30ex}
    $p_i = q_i$ \quad $p_{(\veca,c)} = q_{(\veca, \vec{0}, c)}$\\[1ex]
    $p' = r_{k+1}$ \quad $p'_{(a,c)} = r_{(\vec{0},a,c)}$\\[1ex]
    $\tjudge{\Gamma}{P}{e_1}{A}{P'}$
  \end{minipage}  
  &
  \begin{minipage}[b]{33ex}
   $p^{\vecb}_{(\veca,c)} = q_{(\veca,\vecb,c)}$ \\[1ex]
   ${p'}^{\vecb}_{(a,c)} = r_{(\vecb,a,c)}$\\[1ex]      
   %\quad ${p'}^{\vecb} = 0$
   $\tjudgecf{\Gamma}{P_{\vecb}}{e_1}{A}{P'_{\vecb}} \quad (\vecb \not= \vec{0})$
  \end{minipage}
  &
  \begin{minipage}[b]{25ex}
    $r_{(\vecb,0,c)} = q_{(\vec{0},\vecb,c)}$\\[1ex]
    $r_j = q_j$\\[1ex]
    $\tjudge{\Delta, \typed{x}{A}}{R}{e_2}{C}{Q'}$
  \end{minipage}
  }$}
  \\[3ex]
  $\infer[(\sharerule)]{\tjudge{\Gamma, \typed{z}{A}}{\share{Q}}{e[z,z]}{C}{Q'}}{%
    \tjudge{\Gamma, \typed{x}{A}, \typed{y}{A}}{Q}{e[x,y]}{C}{Q'}
  }$
  &
  $\infer[(\weak)]{%
    \tjudge{\Gamma}{Q}{e}{A}{Q'}
    }{%
    \tjudge{\Gamma}{P}{e}{A}{P'}      
    &
    \begin{minipage}[b]{20ex}
      $\potential{\Gamma}{P} \leqslant \potential{\Gamma}{Q}$
      \\[1ex]
      $\potential{\Gamma}{P'} \geqslant \potential{\Gamma}{Q'}$
    \end{minipage}
    % \potential{\Gamma}{P} \leqslant \potential{\Gamma}{Q}
    % & 
    % \potential{\Gamma}{P'} \geqslant \potential{\Gamma}{Q'}
    }$
  \\[3ex]
  $\infer[(\var)]{\tjudge{\typed{x}{A}}{Q}{x}{A}{Q}}{%
     \text{$x$ a variable}
  }$
  &
  $\infer[(\app)]{%
        \tjudge{\typed{x_1}{A_1},\dots,\typed{x_n}{A_n}}{Q+1}{f(\seq{x})}{A'}{Q'}
      }%
      {%
        \atypdcl{A_1 \times \cdots \times A_n}{Q}{A'}{Q'} \in \FS(f)
    }$
  \end{tabular}  
\end{center}
To ease notation, We set $\veca \defsym \seq[m]{a}$, $\vecb \defsym \seq[k]{b}$, 
$i \in \{1,\dots,m\}$, $j \in \{1,\dots,k\}$ and $a,b,c \in \Qplus$ to simplify notation. 
Sequence elements of annotations, which are not constraint are set to zero.
\caption{Type System for Logarithmic Amortised Resource Analysis}
\label{fig:5}
\end{figure}
 
\begin{definition}
\label{d:welltyped}
A program $\Program$ is called \emph{well-typed} if for any rule
$f(x_1,\dots,x_k) = e \in \Program$ and any annotated signature
$\atypdcl{A_1 \times \cdots \times A_n}{Q}{A'}{Q'} \in \FS(f)$, we have 
$\tjudge{\typed{x_1}{A_1},\dots,\typed{x_k}{A_k}}{Q}{e}{A'}{Q'}$.
A program $\Program$ is called \emph{cost-free} well-typed, if the
cost-free typing relation is employed.
\end{definition}

We obtain the following soundness result.

\begin{theorem}[Soundness Theorem]
\label{t:1}
Let $\Program$ be well-typed and let $\sigma$ be a substitution. Suppose $\tjudge{\Gamma}{Q}{e}{A}{Q'}$ and
$\eval{\sigma}{m}{e}{v}$. Then $\spotential{\Gamma}{Q} - \potential{v}{Q'} \geqslant m$.
\end{theorem}
\begin{proof}
  The proof embodies the high-level description
  given in Section~\ref{Primer}. It proceeds by main induction on $\Pi \colon \eval{\sigma}{m}{e}{v}$
  and by side induction on $\Xi \colon \tjudge{\Gamma}{Q}{e}{A}{Q'}$. 
We consider only a few cases of interest. 
For example, for a case not covered: the variable rule $(\var)$ types a variable of unspecified 
type $A$. As no actual costs a required the annotation is unchanged and the
theorem follows trivially. 

\emph{Case}. $\Pi$ derives $\eval{\sigma}{0}{\nil}{\nil}$. Then $\Xi$ 
consists of a single application of the rule $(\leafrule)$:
\begin{equation*}
\infer[(\leafrule)]{\tjudge{\varnothing}{Q}{\nil}{T}{Q'}}{%
    q_{c} = \sum_{a+b=c} q'_{a,b}
  }  
\tpkt
\end{equation*}
By assumption $Q = [(q_c)_{c \in \N}]$ is an annotation for the empty
sequence of trees. On the other hand $Q' = [(q'_{(a,b)})_{a,b \in \N}]$ is an annotation
of length $1$. 
Thus we obtain:
\begin{equation*}
\potential{\varnothing}{Q} = \sum_{c} q_c \cdot \log(c) =
\sum_{a,b} q'_{(a,b)} \cdot \log(a+b) = 
\rk(\nil) + \sum_{a,b} q'_{(a,b)} p_{(a,b)} = \potential{\nil}{Q'}
\tpkt  
\end{equation*}

\emph{Case}.
Suppose the last rule in $\Pi$ has the following
from:
\begin{equation*}
  \infer{\eval{\sigma}{0}{\tree{x_1}{x_2}{x_3}}{\tree{u}{x}{b}}}{%
      x_1\sigma = u
      &
      x_2\sigma = b
      &
      x_3\sigma = v
  }
\tpkt
\end{equation*}
Wlog.\ $\Xi$ consists of a single application of the rule $\noderule$:
\begin{equation*}
\infer[(\noderule)]{\tjudge{\typed{x_1}{T},\typed{x_2}{B},\typed{x_3}{T}}{Q}{\tree{x_1}{x_2}{x_3}}{T}{Q'}}{%
    q_1 = q_2 = q'_\ast
    &
    q_{(1,0,0)} = q_{(0,1,0)} = q'_\ast
    &
    q_{(a,a,b)} = q'_{(a,b)}
  }  
\end{equation*}
By definition, we have $Q = [q_1, q_2] \cup [(q_{(a_1,a_2,b)})_{a_i,b \in \N}]$ and
$Q' = [q'] \cup [(q'_{(a,b)})_{a',b' \in \N}]$. 
We set $\Gamma \defsym \typed{x_1}{T}, \typed{x_2}{B}, \typed{x_3}{T}$ and 
$\tree{x_1}{x_2}{x_3}\sigma \defsym \tree{u}{b}{v}$. Thus $\spotential{\Gamma}{Q} =
\potential{u,v}{Q}$ and we obtain:
\begin{align*}
\potential{u,v}{Q} & = q_1 \cdot \rk(u) + q_2 \cdot \rk(v) + 
                    \sum_{a_1,a_2,b} q_{(a_1,a_2,b)} \cdot \log(a_1 \cdot \size{u} + a_2 \cdot \size{v} + b)
\\
& \geqslant q'_\ast \cdot \rk(u) + q' \cdot \rk(v) + 
                    q_{(1,0,0)} \cdot \log(\size{u}) + q_{(0,1,0)} \cdot \log(\size{v}) +
\\ 
& \quad  {} + \sum_{a,b} q_{(a,a,b)} \cdot \log(a \cdot \size{u} + a \cdot \size{v} + b)
\\
& = q'_\ast \cdot (\rk(u)+ \rk(v)+ \log(\size{u})+\log(\size{v})) + {}\\
& \quad {} + \sum_{a, b} q'_{(a,b)} \cdot \log(a \cdot (\size{u}+\size{v}) +b) 
\\
& = q'_\ast \cdot \rk(\tree{u}{b}{v}) + \sum_{a, b} q'_{(a,b)} \cdot p_{(a,b)}(\tree{u}{b}{v})
= \potential{\tree{u}{b}{v}}{Q'}
\tpkt
\end{align*}

\emph{Case}.
Consider the $\match$ rule, that is, $\Pi$ ends as follows:
\begin{equation*}
\infer{\eval{\sigma}{m}{\match\ x~\with\
    \text{\lstinline{|}} \nil \arrow e_1
    \text{\lstinline{|}} \tree{x_1}{x_2}{x_3} \arrow e_2
    }{v}}{%
    x\sigma = \tree{t}{a}{u}
    &
    \eval{\sigma'}{m}{e_2}{v}  
  }
\tpkt  
\end{equation*}
Wlog.\ we may assume that $\Xi$ ends with the related application
of the $(\matchrule)$:
\begin{equation*}
\infer[(\matchrule)]{\tjudge{\Gamma, \typed{x}{T}}{Q}{\match\ x~\with\ 
    \text{\lstinline{|}} \nil \arrow e_1
    \text{\lstinline{|}} \tree{x_1}{x_2}{x_3} \arrow e_2}{A}{Q'}}{%
    \begin{minipage}[b]{30ex}
      $r_{(\veca,a,a,b)} = q_{(\veca,a,b)}$\\[1ex]
      $p_{\veca,c} = \sum_{a+b=c} q_{\veca,a,b}$\\[1ex]
      $\tjudge{\Gamma}{P}{e_1}{A}{Q'}$
    \end{minipage}
  &
    \begin{minipage}[b]{37ex}
       $r_{m+1} = r_{m+2} = q_{m+1}$
      \\[1ex]
      $r_{(\vec{0},1,0,0)} = r_{(\vec{0},0,1,0)} = q_{m+1}$
      \\[1ex]
      $\tjudge{\Gamma, \typed{x_1}{T}, \typed{x_2}{B}, \typed{x_3}{T}}{R}{e_2}{A}{Q'}$
    \end{minipage}
  }
\tpkt
\end{equation*}

We assume the annotations $P$, $Q$, $R$, are of length $m$, $m+1$ and $m+3$,
respectively, while $Q'$ is of length $1$. We write $\vect \defsym \seq[m]{t}$ for
the substitution instances of the variables in $\Gamma$ and $t \defsym x\sigma = \tree{u}{b}{v}$,
where the latter equality follows from the assumption on $\Pi$.
By definition and the constraints given in the rule, we obtain:
\begin{align*}
%\spotential{v}{Q'} & q' \rk(v) + \sum_{a,c} q' \log(a\size{v} + c)
\spotential{\sigma}{\Gamma,\typed{x}{T}}{Q} &= \sum_{i} q_i \rk(t_i) + q_{m+1} \rk(t) +
                                              \sum_{\veca, a, c} \log(\veca\size{\vect} + c)
\\
& = \sum_{i} q_i \rk(t_i) + q_{m+1} \rk(\tree{u}{b}{v}) + 
\sum_{\veca, a, c} q_{(\veca, a, c)} p_{(\veca, a, c)} (\vect, t) 
\\
&=
\sum_{i} q_i \rk(t_i) + q_{m+1}(\rk(u)+\log(\size{u})+\log(\size{v})+\rk(v)) + {}
\\
&\phantom{=}
{} + \sum_{\veca, a, c} \log(\veca \size{\vect} + a(\size{u}+\size{v}) + c)
\\
&= \spotential{\sigma}{\Gamma, \typed{x_1}{T}, \typed{x_2}{B}, \typed{x_3}{T}}{R}
\tkom
\end{align*}
where, we shortly write $\veca\size{\vect}$ ($\vecb\size{\vecu}$) to
denote componentwise multiplication.

Thus $\spotential{\sigma}{\Gamma,\typed{x}{T}}{Q} = \spotential{\sigma}{\Gamma, \typed{x_1}{T}, \typed{x_2}{B}, \typed{x_3}{T}}{R}$ and the theorem follows by an application of MIH.

\emph{Case}.
Consider the $\vlet$ rule, that is, $\Pi$ ends in the following rule:
\begin{equation*}
  \infer{\eval{\sigma}{m}{\vlet\ x~\equal~e_1\ \vin\ e_2}{v}}{%
     \eval{\sigma}{m_1}{e_1}{v'}
     &
     \eval{\sigma[x \mapsto v']}{m_2}{e_2}{v}
  }
  \tkom
\end{equation*}
where $m = m_1+m_2$. Wlog.\ $\Xi$ ends
in the following application of the $\letrule$-rule.
\begin{equation*}
  \infer{\tjudge{\Gamma, \Delta}{Q}{\vlet\ x~\equal~e_1\ \vin\ e_2}{C}{Q'}}{%
  \tjudge{\Gamma}{P}{e_1}{T}{P'}
  &
  \tjudgecf{\Gamma}{P_{\vecb}}{e_1}{T}{P'_{\vecb}}
  &
  \tjudge{\Delta, \typed{x}{T}}{R}{e_2}{C}{Q'}
  }
  \tpkt
\end{equation*}
Recall that $\veca = \seq[m]{a}$, $\vecb = \seq[k]{b}$, 
$i \in \{1,\dots,m\}$, $j \in \{1,\dots,k\}$ and $a,c \in \Qplus$.
Further, the annotations $Q$, $P$, $R$ are annotation of length $m+k$, $m$ and
$k$, respectively,  while $Q',P',R'$ are of length $1$. For each sequence 
$\seq[k]{b} \not= \vec{0}$, $P_{\vecb}$ denotes an annotation of length $m$.
We tacitly assume that $\veca \not = \vec{0}$, as well as $\vecb \not = \vec{0}$.
Furthermore, recall the convention that
the sequence elements of annotations are denoted by the lower-case letter of the
annotation, potentially with corresponding sub- or superscripts.

By definition and due  to the constraints expressed in the typing rule, we have
\begin{align*}
  \spotential{\Gamma,\Delta}{Q} & = 
  \sum_i q_i \rk(t_i) + \sum_j q_j \rk(u_j) +
  \sum_{\veca,\vecb,c} q_{(\veca,\vecb,c)} \log(\veca\size{\vect} + \vecb\size{\vecu} + c)
  \\
  \spotential{\Gamma}{P} & = \sum_i q_i \rk(t_i) + \sum_{\veca,c} q_{(\veca, \vec{0}, c)} 
\log (\veca \size{\vect} + c)
  \\
  \potential{v'}{P'} &= r_{k+1} \rk(v') + \sum_{a,c} r_{(\vec{0},a,c)} \log (a\size{v} +c)
  \\
  \spotential{\Gamma}{P_{\vecb}} &= \sum_{\veca,c} q_{(\veca,\vecb,c)} \log (\veca \size{\vect} + c)
  \\
  \potential{v'}{{P'}_{\vecb}} &= \sum_{a,c} r_{(\vec{b},a,c)} \log (a\size{v} +c)
  \\
  \spotential{\Delta,\typed{x}{T}}{R} & = \sum_{j} q_j \rk(u_j) + r_{k+1} \rk(v') +
\sum_{\vecb,a,c} r_{(\vecb,a,c)} \log(\vecb\size{\vecu}+a\size{v}+c)
  \tkom
\end{align*}
where we set $\vect \defsym \seq[m]{t}$ and $\vecu \defsym \seq[k]{u}$, denoting
the substitution instances of the variables in $\Gamma$, $\Delta$, respectively.

% $p_i = q_i$, $p_{(\vec{a},c)} = q_{(\overline{a}, \overline{0}, c)}$
% $p' = r_{k+1}$, $p'_{(a,c)} = r_{(\overline{0},a,c)}$
% $p^{\overline{b}}_{(\overline{a},c)} = q_{(\overline{a},\overline{b},c)}$
% ${p'}^{\overline{b}}_{(a,c)} = r_{(\overline{b},a,c)}$
% ${p'}^{\overline{b}} = 0$
% $r_{(\overline{b},0,c)} = q_{(\overline{0},\overline{b},c)}$
% $r_j = q_j$%

By main induction hypothesis, we conclude that
$\spotential{\Gamma}{P} - \potential{v'}{P'} \geqslant m_1$, while for all
$\vecb \not= \vec{0}$, $\spotential{\Gamma}{P^{\vecb}} \geqslant \potential{v}{{P'}^{\vecb}}$.
A second application of MIH yields that 
$\spotential{\Delta,\typed{x}{T}}{R} - \potential{v}{Q'} \geqslant m_2$. 
Due to Lemma~\ref{l:3}, we can combine these two results and
conclude the theorem. 
\end{proof}

% \begin{verbatim} 
% Hi  Georg, ich glaube  uebrigens, dass, wenn  man als
% Basic  Potential  Funktion  neben  rk(x) und  log(ax+b)  auch  noch  x
% (lineare Funktion)  erlaubt, man dann auch  mergesort analysieren kann
% und zwar so,  dass n log n rauskommt (also  a*rk(x)+b). Das waere doch
% dann wirklich eine schoene Story. --martin
%  \end{verbatim}
% TODO: lineares potential in TR formulieren und mergesort skizzieren

\begin{remark}
  As remarked in Section~\ref{ResourceFunctions} the basic resource functions can be generalised
  to additionally represent linear functions in the size of the arguments. The above soundness theorem is
  not affected by this generalisation.
\end{remark}

\section{Analysis}
\label{Analysis}

In this section, we exemplify the use of the type system presented in the last section
on the function \lstinline{splay}, cf.~Figure~\ref{fig:1}. 
Our amortised analysis of splaying yields that the amortised cost
of \lstinline{splay a t} is bound by $1 + 3\log(\size{t})$, where the actual
cost count the number of calls to \lstinline{splay}, cf.~\cite{ST:1985,Schoenmakers93,Nipkow:2015}.
To verify this declaration, we derive
\begin{equation}
  \label{eq:splay:1}
  \tjudge{\typed{a}{B},\typed{t}{T}}{Q}{e}{T}{Q'} \tkom
\end{equation}
where the expression $e$ is the definition of \lstinline{splay} given in
Figure~\ref{fig:1}. We restrict to the zig-zig case: $t = \tree{\tree{bl}{b}{br}}{c}{cr}$
together with the recursive call \lstinline[mathescape]{splay a bl = $\tree{al}{a'}{ar}$} and
$a < b < c$. Thus \lstinline{splay a t} yields $\tree{al}{a'}{\tree{ar}{b}{\tree{br}{c}{cr}}} =: t'$.
Recall that $a$ need not occur in $t$, in this case the last element $a'$ before
a leaf was found is rotated to the root.

Let $e_1$ denote the subexpression of the definition of splaying, starting in
program line $4$. On the other hand let $e_2$ denote the subexpression defined from
line $5$ to $15$ and let $e_3$ denote the program code within $e_2$ starting
in line $7$. Finally the expression in lines $11$ and $12$, 
expands to the following, if we remove part of the the syntactic sugar.
\begin{equation*}
  e_4 \defsym \vlet\ x\ \equal\ \text{\lstinline{splay a bl}}\ \vin\ 
    \match\ x\ \with\ \text{\lstinline{|}} \nil\ \arrow \nil\
    \text{\lstinline{|}} \tree{al}{a'}{ar}\ \arrow t'
    \tpkt
\end{equation*}
% $e_4 \defsym \vlet\ x\ \equal\ \text{\lstinline{splay a bl}}\ \vin\ 
% \match\ x\ \with\ \text{\lstinline{|}} \nil\ \arrow \nil\
% \text{\lstinline{|}} \tree{al}{a'}{ar}\ \arrow t'$.

\begin{figure}
\centering
\begin{equation*}
\infer={\tjudge{\typed{a}{B},\typed{t}{T}}{Q}{\match\ t\ 
    \with \text{\lstinline{|}} \nil\  \arrow\ \nil
    \text{\lstinline{|}} \tree{cl}{c}{cr}\ \arrow\ e_1}{T}{Q'}}{%
  \infer{\tjudge{\typed{a}{B},\typed{cl}{T},\typed{c}{B},\typed{cr}{T}}{Q_1}{\cif\ a=c\ 
    \cthen\ \tree{cl}{c}{cr}\ \celse\ e_2}{T}{Q'_1}}{%
  \infer={%
     \tjudge{\typed{a}{B},\typed{b}{B},\typed{cl}{T},\typed{cr}{T}}{Q_1}{\match\ cl\ \with\
     \text{\lstinline{|}} \nil\ \arrow\ \tree{cl}{c}{cr}
     \text{\lstinline{|}} \tree{bl}{b}{br}\ \arrow\ e_3}{T}{Q'}}{%
     \infer[(\weak)]{\tjudge{\Gamma,\typed{cr}{T},\typed{bl}{T},\typed{br}{T}}{Q_2}{e_3}{T}{Q'}}{%
       \infer={\tjudge{\Gamma,\typed{cr}{T},\typed{bl}{T},\typed{br}{T}}{Q_3}{e_4}{T}{Q'}}{%
         \infer{\tjudge{\typed{a}{B},\typed{bl}{T}}{Q+1}{\text{\lstinline{splay a bl}}}{T}{Q'}}{%
           \typed{f}{ \atypdcl{A_1 \times \cdots \times A_n}{Q}{A'}{Q'}}}
         &
         \infer={\tjudge{\Delta,\typed{x}{T}}{Q_4}{\match\ x\ \with\ 
             \text{\lstinline{|}} \tree{al}{a'}{ar}\ \arrow t'}{T}{Q'}}{%
            \tjudge{\Delta,\typed{al}{T},\typed{a'}{B},\typed{ar}{T}}{Q_5}{t'}{T}{Q'}
         }
       }
     }
   }
  }
}
\end{equation*}
\caption{Partial Typing Derivation for \lstinline{splay}, Focusing on the Zig-Zig Case.}
\label{fig:6}
\end{figure}

Figure~\ref{fig:6} shows a simplified derivation of~\eqref{eq:splay:1}, where
we  have focused only on a particular path in the derivation tree, suited
to the assumption on $t$. Omission of premises is indicated by double lines in the inference
step. We abbreviate
$\Gamma \defsym \typed{a}{B},\typed{b}{B},\typed{c}{C}$,
$\Delta \defsym \typed{b}{B},\typed{c}{C},\typed{cr}{T},\typed{br}{T}$.
Here, we use the following annotations, induced by constraints in the type system, cf.~Figure~\ref{fig:5}.
\begin{align*}
  Q &\colon q=1, q_{1,0} = 3, q_{0,2} = 1 \tkom
  \\
  Q' &\colon q'_\ast =1 \tkom
  \\
  Q_1 &\colon q^1_1 = q^1_2 = q = 1, q^1_{1,1,0} = q_{1,0} = 3,
      q^1_{1,0,0} = q^1_{0,1,0} = q = 1,
      q^1_{0,0,2} = q_{0,2} = 1 \tkom
  \\
  Q_2 &\colon q^2_1 = q^2_2 = q^2_3 = 1,
        q^2_{0,0,2} = 1,
        q^2_{1,1,1,0} = q^1_{1,1,0} = 3, 
        q^2_{0,1,1,0} = q^1_{1,0,0} = 1,
  \\
      & \quad 
        q^2_{1,0,0,0} = q^1_{0,1,0} = 1,
        q^2_{0,1,0,0} = q^2_{0,0,1,0} = q^1_1 = 1 \tkom
  \\
  Q_3 &\colon q^3_1=q^3_2=q^3_3 =1,
        q^3_{0,0,2} = 2, 
        q^3_{0,1,0,0} = 3, q^3_{0,0,1,0} = 1, q^3_{1,0,0,0} = q^3_{1,0,1,0} = q^3_{1,1,1,0} =1 \tpkt
\end{align*}
We emphasise that a simple symbolic calculation, following the heuristics
outlined on page~\pageref{d:heuristics} suffices to conclude the following
inequality, employed in the indicated weakening step in Figure~\ref{fig:6}.
\begin{equation*}
  \potential{\Gamma,\typed{cr}{T},\typed{bl}{T},\typed{br}{T}}{Q_2} \geqslant 
  \potential{\Gamma,\typed{cr}{T},\typed{bl}{T},\typed{br}{T}}{Q_3} \tpkt
\end{equation*}
We verify the correctness of the weakening step through a direct comparision.
Let $\sigma$ be a substitution. Then, we have
\begin{align*}
  \spotential{\typed{cr}{T},\typed{bl}{T},\typed{br}{T}}{Q_2} & =
  1 + \rk(cr)+\rk(bl)+\rk(br) + {} \\
  & \quad {} + 3\log(\size{cr}) + 3\log(\size{bl}) + 3\log(\size{br}) +  {}\\
  & \quad {} +    \log(\size{bl}+\size{br}) + \log(\size{cr}) + \log(\size{bl}) + \log(\size{br})      
  \\[1ex]
  & =
  1 + \rk(cr)+\rk(bl)+\rk(br) + 2\log(\size{t}) + \log(\size{t}) + {}
  \\
  & \quad {} + \log(\size{bl}+\size{br}) + \log(\size{cr}) + \log(\size{bl}) + \log(\size{br})
  \\[1ex]
  & \geqslant 
  1 + \rk(cr)+\rk(bl)+\rk(br) +  \log(\size{bl}) + \log(\size{br}) + {}
  \\
  & \quad {} + \log(\size{bl}+\size{br}) + \log(\size{cr}) +  \log(\size{bl}) + {}
  \\
  & \quad {} + \log(\size{br}+\size{cr}) + 2 + \log(\size{bl}+\size{br}+\size{cr})
  \\[1ex]
  & \geqslant \rk(bl) + 1+ 3\log(\size{bl}) +  \rk(cr)+\rk(br)+ \log(\size{br}) + {}
  \\ 
  & \quad {} + \log(\size{cr}) + \log(\size{br}+\size{cr}) + {}
  \\
  & \quad {} + \log(\size{bl}+\size{br}+\size{cr}) + 1
  =  \spotential{\typed{cr}{T},\typed{bl}{T},\typed{br}{T}}{Q_3}
  \tpkt
\end{align*}
Note that we have used Lemma~\ref{l:1} in the third line to conclude
\begin{equation*}
  2 \log(\size{t}) \geqslant \log(\size{bl}) + \log(\size{br}+\size{cr}) + 2 \tkom
\end{equation*}
%
%$2 \log(\size{t}) \geqslant \log(\size{bl}) + \log(\size{br}+\size{cr}) + 2$,
as we have $\size{t} = \size{\tree{\tree{bl}{b}{br}}{c}{cr}} = \size{bl} + \size{br} + \size{cr}$.
Furthermore, we have only used monotonicity of $\log$ and formal simplifications.
In particular all necessary steps are covered in the simple heurstics introduced in Section~\ref{Typesystem}.

Furthermore, the $(\letrule)$-rule is applicable 
with respect to the following annotation $Q_4$:
\begin{align*}
  Q_4 & \colon q^4_1 = q^4_2 = q^4_3 = 1, q^4_{1,0,0,0} = q^4_{0,1,0,0} = q^4_{1,1,0,0} = q^4_{1,1,1,0} = 1 \tpkt
\end{align*}
It suffices to verify the cost-free typing relation
\begin{equation}
\label{eq:splay:2}
\tjudgecf{\typed{a}{B}, \typed{bl}{T}}{P_{\vecb}}{\text{\lstinline{splay a bl}}}{T}{P'_{\vecb}}
\tkom  
\end{equation}
where $\vecb = (b_1,b_2) \not= \vec{0}$. Note that the condition~\eqref{eq:splay:2} has been omitted
from Figure~\ref{fig:6} to allow for a condensed presentation.
Informally speaking~\eqref{eq:splay:2}
requires that in a cost-free computation the potential is preserved. The interesting sub-case is
the case for $\vecb = (1,1)$, governed by the annotations $P_{1,1}$ and $P'_{1,1}$, respectively.
%
% \begin{align*}
%   P_{1,1} & p^{1,1}_{0,0} = 1, p^{1,1}_{1,0} =1
% \\
%   P'_{1,1} & {p'}^{1,1}_{0,0} = 1, {p'}^{1,1}_{1,0} =1
% \end{align*}
%
The corresponding potentials are $\spotential{\typed{a}{B}, \typed{bl}{T}}{P_{1,1}} = \log(\size{bl})$
and $\spotential{\typed{a}{B}, \typed{bl}{T}}{P'_{1,1}} = \log(\size{\tree{al}{a'}{ar}})$. 
As $\size{bl} = \size{\tree{al}{a'}{ar}}$ by definition of \lstinline{splay}, the potential remains unchanged
as required.

Finally, one further application of the $\matchrule$-rule yields the desired 
derivation for suitable~$Q_5$. 

\section{Towards Automatisation}
\label{Implementation}

In this short section, we argue that the above introduced potential-based amortised resource
analysis is automatable. 
As emphasised in Section~\ref{Primer} the principal approach to
automatisation is to set up annotations with indeterminate coefficients and solve
for them so as to automatically infer costs. The corresponding constraints are obtained
through a syntax-directed type inference. 
In the context of the type system presented in Figure~\ref{fig:5} an obvious challenge is the
requirement to compare potentials symbolically (compare~Section~\ref{ResourceFunctions}) rather
than compare annotations directly. 
More generally, the presence of logarithmic basic functions necessitates
the embodiment of nonlinear arithmetic.

A straightforward approach for automation would exploit recent advances in SMT solving. For this
one can suitable incorporating the required nonlinear arithmetic as axioms to an off-the-shelf
solver and pass the constraints to the solver. We have experimented with this approach, but the approach
has turned out to be too inefficient. In particular, as we cannot enforce \emph{linear} constraints.

However, a more refined and efficient approach which targets linear constraints is achievable as follows.
All logarithmic terms, that is, terms of the form $\log(.)$ are replaced by new variables, focusing on finitely many.
For the latter we exploit the condition that in resource annotation only finitely many coefficients are non-zero.
Wrt.\ the example in the previous section, $\log(\size{bl}+\size{br})$, $\log(\size{bl})$ are replaced by
the fresh (constraint) variables $x$, $y$, respectively.
Thus laws of the monotonicity function, like e.g.\ monotonicity of $\log$, as well as properties like Lemma~\ref{l:1} 
can be expressed as inequalities over the introduced unknowns. E.g., the inequality $x \geqslant y$ represents
the axiom of monotonicity $\log(\size{bl}+\size{br}) \geqslant \log(\size{bl})$. All such obtained inequality
are collected as ``expert knowledge''.
We can express the required expert knowledge succinctly in the form of a system of inequalities
as $Ax \leqslant b$, where $A$ denotes a matrix with as many rows as we have expert knowledge,
$\vec{b}$ a column vector and $\vec{x}$ 
the column vector of unknowns of suitable length. 
With the help of the variables in $\vec{x}$, we construct linear combinations based on indeterminate coefficients
giving rise to the potential functions fulfilling the constraints gathered from type inference. More precisely,
we have to solve the implication
%$\forall \vec{x}\ A\vec{x} \leqslant \vec{b} \Rightarrow CY\vec{x} \leqslant \vec{d}$.
%
\begin{equation}
\label{eq:constraint}
\forall \vec{x}\ A\vec{x} \leqslant \vec{b} \Rightarrow CY\vec{x} \leqslant \vec{d}
\tpkt   
\end{equation}
Here $C$ is the matrix of coefficients and the matrix $Y$ represents the
required linear combinations.

In order to automate the derivation of~\eqref{eq:constraint}, we exploit the following
variant of Farkas' Lemma. 

\begin{lemma}
\label{l:farkas}  
  Suppose  $A\vec{x} \leqslant \vec{b}$ is solvable. Then the following assertions are
  equivalent.
  \begin{align}
    \forall \vec{x}\ A\vec{x} \leqslant \vec{b} \Rightarrow \vec{u}^T\vec{x} \leqslant \lambda
    \label{eq:farkas:1}
    \\
    \exists \vec{f} \ f \geqslant 0 \land \vec{u}^T \leqslant \vec{f}^T A \land \vec{f}^T \vec{b} \leqslant \lambda
    \label{eq:farkas:2}
  \end{align}
\end{lemma}
\begin{proof}
  It is easy to see that from \eqref{eq:farkas:2}, we obtain \eqref{eq:farkas:1}. Assume~\eqref{eq:farkas:2}. Assume further that
  $A\vec{x} \leqslant \vec{b}$ for some column vector $\vec{x}$. Then we have
  \begin{equation*}
    \vec{u}^T\vec{x} \leqslant \vec{f}^T A \vec{x} \leqslant \vec{f}^T \vec{b} \leqslant \lambda
    \tpkt
  \end{equation*}
  Note that for this direction the assumption that $A\vec{x} \leqslant \vec{b}$ is solvable is not required.

  With respect to the opposite direction, we assume~\eqref{eq:farkas:1}. By assumption, the inequality $A\vec{x} \leqslant \vec{b}$
  is solvable. Hence, maximisation of $\vec{u}^T \vec{x}$ under the side condition  $A\vec{x} \leqslant \vec{b}$ is feasible. Let
  $w$ denote the maximal value. Due to~\eqref{eq:farkas:1}, we have $w \leqslant \lambda$.

  Now, consider the dual asymmetric
  linear program to minimise  $\vec{y}^T \vec{b}$ under side condition $\vec{y}^T A = \vec{u}^T$ and $\vec{y} \geqslant 0$. Due to
  the Dualisation Theorem, the dual problem is also solvable with the same solution
  \begin{equation*}
    \vec{y}^T \vec{b} = \vec{u}^T \vec{x} = w \tpkt
  \end{equation*}
  We fix a vector $\vec{f}$ which attains the optimal value $w$, such that $\vec{f}^T A = \vec{u}^T$ and $\vec{f} \geqslant 0$
  such that $\vec{f}^T \vec{b} = w \leqslant lambda$. This yields~\eqref{eq:farkas:2}. 
\end{proof}

Generalising Lemma~\ref{l:farkas}, we obtain the following equivalence, which allows an
efficient encoding of~\eqref{eq:constraint}. 
%
% $\forall \vec{x}\ A\vec{x} \leqslant b \Rightarrow U\vec{x} \leqslant \vec{v}$
% iff $\exists F \geqslant 0 \land U \leqslant FA \land F\vec{b} \leqslant \vec{v}$,
%
\begin{equation*}
\forall \vec{x}\ A\vec{x} \leqslant \vec{b} \Rightarrow U\vec{x} \leqslant \vec{v} \Leftrightarrow
\exists F \geqslant 0 \land U \leqslant FA \land F\vec{b} \leqslant \vec{v}  
\tpkt
\end{equation*}
As in the lemma, the equivalence requires solvability of the system $A\vec{x} \leqslant \vec{b}$. 
Note that the system expresses given domain knowledge and simple facts like Lemma~\ref{l:1}, whose solvablity is given a priori.
We emphasise that the existential statement requires linear constraints only.

\section{Conclusion}
\label{Conclusion}

We have presented a novel amortised resource analysis based on the potential method. The
method is rendered in a type system, so that resource analysis amounts to a constraint satisfaction
problem, induced by type inference. 
The novelty of our contribution is that this is the first automatable
approach to logarithmic amortised complexity. In particular, we show how the precise logarithmic complexity
of \emph{splaying}, a central operation of Sleator and Tarjan's splay trees can be analysed in our system. 
Furthermore, we provide a suitable \emph{Ansatz} to automatically infer logarithmic bounds on the
runtime complexity.

\medskip
\paragraph*{In Memorium.} With deep sorrow, I report that Martin had a fatal hiking accident 
during the preparation of this work. He passed away in January, 2018. 
I've tried my best to finalise our common conceptions and ideas, any mistakes or other defects introduced are of course my responsibility.
His work was revolutionary in a vast amount of fields and it
will continue to inspire future researchers; like he inspired me. 

%\bibliography{references}

\end{document}